\documentclass[4paper,envcountsame,envcountreset,envcountsect]{llncs}
\usepackage[utf8]{inputenc}
\usepackage{algorithm,xspace}
\usepackage{algpseudocode}
\usepackage{ dsfont }
\usepackage{amsmath}
\usepackage{amsfonts}
\usepackage{tikz}
\usetikzlibrary{matrix}
\usetikzlibrary{arrows}
\usepackage{appendix}
\usepackage{graphicx,wrapfig}
\usepackage{ amssymb }
\usepackage{hyperref}
\usepackage{color}

\usepackage{perpage}
\MakePerPage{footnote}
\usepackage{xcolor}

\newcommand{\UPDP}{\textsc{UpPlan-VDPP}\xspace} 
\newcommand{\Oof}{\ensuremath{\mathcal{O}}}

\newcommand*{\R}{\ensuremath{\mathbb{R}}}
\newcommand*{\N}{\ensuremath{\mathbb{N}}}

\newcommand{\PPP}{\mathcal{P}}
\newcommand{\CCC}{\mathcal{C}}

\newcommand{\Right}[1]{\textit{right}(#1)}
\newcommand{\Left}[1]{\textit{left}(#1)}

\newcommand{\npprob}[4]{
\begin{center}\normalfont\fbox{%
\begin{tabular}[t]{rp{#1}}%
\multicolumn{2}{l}{\textsc{#2}}\\%
\textit{Input:} & #3.\\%
\textit{Problem:} & #4%
\end{tabular}}%
\end{center}
}

\spnewtheorem{fact}{Fact}{\bfseries}{\itshape}

\renewenvironment{proof}[1][]{\noindent{%
\ifthenelse{\equal{#1}{}}{{\sl Proof.\ }}{{\sl #1.\ }}%
}}{\hspace{1em}\nobreak\hfill$\Box$\endtrivlist\addvspace{2ex plus
0.5ex minus0.1ex}}

\author{Saeed Akhoondian Amiri \inst{1} \and Ali
 Golshani\inst{2} \\ \and Stephan Kreutzer\inst{1} \and Sebastian
  Siebertz \inst{1}}
\institute{\small Technical University
  Berlin \email{saeed.akhoondianamiri,stephan.kreutzer,sebastian.siebertz@tu-berlin.de}
  \and University of Tehran, \email{ali.golshani@ut.ac.ir}}

\begin{document}

\title{Vertex Disjoint Paths in Upward Planar Graphs}

\maketitle

\begin{abstract}
  The $k$-vertex disjoint paths problem is one of the most studied
  problems in algorithmic graph theory. In 1994, Schrijver proved that
  the problem can be solved in polynomial time for every fixed $k$
  when restricted to the class of planar digraphs and it was a long
  standing open question whether it is fixed-parameter tractable (with
  respect to parameter $k$) on this restricted class. Only recently,
  \cite{CMPP}.\ achieved a major breakthrough and answered the
  question positively. Despite the importance of this result (and the
  brilliance of their proof), it is of rather theoretical
  importance. Their proof technique is both technically extremely
  involved and also has at least double exponential parameter
  dependence. Thus, it seems unrealistic that the algorithm could
  actually be implemented.  In this paper, therefore, we study a
  smaller class of planar digraphs, the class of upward planar
  digraphs, a well studied class of planar graphs which can be drawn
  in a plane such that all edges are drawn upwards. We show that on
  the class of upward planar digraphs the problem (i) remains
  NP-complete and (ii) the problem is fixed-parameter tractable. While
  membership in FPT follows immediately from \cite{CMPP}'s general
  result, our algorithm has only single exponential parameter
  dependency compared to the double exponential parameter dependence
  for general planar digraphs. Furthermore, our algorithm can easily
  be implemented, in contrast to the algorithm in \cite{CMPP}.
\end{abstract}

\section{Introduction}
Computing vertex or edge disjoint paths in a graph connecting given
sources to sinks is one of the fundamental 
problems in algorithmic graph theory with applications in VLSI-design,
network reliability, routing and many other areas. There are many
variations of this problem which differ significantly in their
computational complexity. If we are simply given a graph (directed or
undirected) and two sets of vertices $S, T$ of equal cardinality, and
the problem is to compute $|S|$ pairwise vertex or
edge disjoint paths connecting sources in $S$ to targets in $T$, then this problem
can be solved efficiently by standard network flow techniques. 

A variation of this is the well-known \emph{$k$-vertex disjoint paths
  problem}, where the sources and targets are given as lists $(s_1, \dots,
s_k)$ and $(t_1, \dots, t_k)$ and the problem is to find $k$ vertex
disjoint paths connecting each source $s_i$ to its corresponding
target $t_i$. The $k$-disjoint paths problem is NP-complete in general
and remains NP-complete even on planar undirected graphs (see \cite{GareyJ79}).

On undirected graphs, it can be solved in polynomial time for any
fixed number $k$ of source/target pairs. This was first proved for the
$2$-disjoint paths problems, for instance in
\cite{Seymour80,Shiloach80,Thomassen80,Ohtsuki81},  before
Robertson and Seymour proved in \cite{GMXIII} that the problem can be solved in
polynomial-time for every fixed $k$. In fact, they proved more,
namely that the problem is \emph{fixed-parameter tractable} with
parameter $k$, that is, solvable in time $f(k)\cdot |G|^c$, where $f$
is a computable function, $G$ is the input graph, $k$ the number of
source/target pairs and $c$ a fixed constant (not depending on
$k$). See e.g.~\cite{DowneyF98} for an introduction to fixed-parameter tractability.

For directed graphs the situation is quite different
(see~\cite{BangJensenG10} for a survey). Fortune et
al. \cite{FortuneHW80} proved that the problem is already NP-complete
for $k=2$ and hence the problem is not fixed-parameter tractable on
directed graphs. It is not even fixed-parameter tractable on acyclic
digraphs, as shown by Slivkins \cite{Slivkins03}. However, on acyclic
digraphs the problem can be solved in polynomial time for any fixed
$k$ \cite{FortuneHW80}. 

In \cite{JohnsonRobSeyTho01}, Johnson et al. introduced the concept of
\emph{directed tree-width} as a directed analogue of undirected
tree-width for directed graphs. They showed that on classes of
digraphs of bounded directed tree-width the $k$-disjoint
paths problem can be solved in polynomial time for any fixed $k$. 
As the class of acyclic digraphs has directed tree-width $1$,
Slivkins' result \cite{Slivkins03} implies that the
problem is not fixed-parameter tractable on such classes.

Given the computational intractability of the directed disjoint paths
problem on many classes of digraphs, determining classes of digraphs
on which the problem does become at least fixed-parameter tractable is
an interesting and important problem. Using colour coding techniques,
the problem can be shown to become fixed-parameter tractable if the
length of the disjoint paths is bounded. This has, for instance, been
used to show fixed-parameter tractability of the problem on classes of
bounded \emph{DAG-depth} \cite{GanianHKLOR09}. In 1994, Schrijver
\cite{Schrijver94} proved that the directed $k$-disjoint paths problem
can be solved in polynomial time for any fixed $k$ on planar digraphs,
using a group theoretical approach and it was a long standing open
question whether it is fixed-parameter tractable on this restricted
class.  Only recently, Cygan et al.\ achieved a major breakthrough and
answered the question positively. Despite the importance of this
result (and the brilliance of their proof), it is of rather
theoretical importance. Their proof technique is based on irrelevant
vertices in directed grids and both technically extremely involved and
also has at least double exponential parameter dependence. Thus, it
seems unrealistic that the algorithm could actually be implemented.

In this paper, therefore, we study a smaller class of planar digraphs,
the class of \emph{upward planar digraphs}.  These are graphs that
have a plane embedding such that every directed edge points
``upward'', i.e.~each directed edge is represented by a curve that is
monotone increasing in the $y$ direction.  Upward planar digraphs are
very well studied in a variety of settings, in particular in graph
drawing applications (see e.g.~\cite{BattistaETT99}). In contrast to
the problem of finding a planar embedding for a planar graph, which is
solvable in linear time, the problem of finding an upward planar
embedding is NP-complete in general \cite{UpwardEmbedding}. Much work
has gone into finding even more restricted classes inside the upward
planar class that allow to find such embeddings in polynomial time
\cite{SDUpward,TriconnectedUpward,outerplanardagupward}.
  
By definition, upward planar graphs are planar graphs. Hence, by the
above results, the $k$-vertex disjoint paths problem can be solved in
polynomial time on upward planar graphs for any fixed $k$. As a first
main result in this paper we show that the problem remains NP-complete
on upward planar graphs, i.e.,~that this cannot be improved to a
general polynomial-time algorithm. Our construction even shows that
the problem is NP-complete on directed grid graphs.

Our second main result is that the problem is fixed-parameter
tractable with respect to parameter $k$ on the class of upward planar
digraphs if we are given an upward planar graph together with an
upward planar embedding. We present a linear time algorithm
that has single exponential parameter dependency.

\section{Preliminaries}

By $\N$ we denote the set of non-negative integers and for $n\in\N$,
we write $[n]$ for the set $\{1,\ldots, n\}$. We assume familiarity
with the basic concepts from (directed) graph theory, planar graphs
and graph drawings and refer the reader to
\cite{BangJensenG10,BattistaETT99,Diestel05} for more details. For
background on parameterized complexity theory we refer the reader to
\cite{DowneyF98}.

An \emph{upward planar graph} is a graph that has a plane embedding
such that every directed edge points ``upward'', i.e.~each directed
edge is represented by a curve that is monotone increasing in the $y$
direction.

The $k$-vertex disjoint paths problem on upward planar graphs is the
following problem.
\vspace{-3mm}
\npprob{10cm}{Vertex Disjoint Paths on Upward Planar Graphs (\UPDP)}%
{An upward planar graph $G$ together with an upward planar embedding,
  $(s_1, t_1), \dots, (s_k, t_k)$}%
{Decide whether there are $k$ pairwise internally vertex disjoint paths $P_1,
  \dots, P_k$ linking $s_i$ to $t_i$, for all $i$.}

\section{NP-Completeness of \UPDP}
\label{sec:hardness}
This section is dedicated to the proof of one of our main theorems:

\begin{theorem}\label{thm:np}
  \UPDP is NP-complete.
\end{theorem}
Before we formally prove the theorem, we give a brief and informal
overview of the proof structure.  The proof of NP-completeness is by a
reduction from SAT, the satisfiability problem for propositional
logic, which is well-known to be NP-complete \cite{GareyJ79}.  On a
high level, our proof method is inspired by the NP-completeness proof
in \cite{Lynch75} but the fact that we are working in a restricted
class of planar digraphs requires a number of changes and additional
gadgets.

Let $\mathcal{V}=\{V_1,\ldots, V_n \}$ be a set of variables and
$\mathcal{C}=\{C_1,…,C_m \}$ be a set of clauses over the variables
from $\mathcal{V}$.  For $1\leq i\leq m$ let
$C_i=\{L_{i,1},L_{i,2},…,L_{i,n_i}\}$ where each $L_{i,t}$ is a
literal, i.e.,~a variable or the negation thereof. We will construct
an upward planar graph $G_\mathcal{C} = (V, E)$ together with a set of
pairs of vertices in $G_\mathcal{C}$ such that $G_\mathcal{C}$
contains a set of pairwise vertex disjoint directed paths connecting
each source to its corresponding target if, and only if, $\mathcal{C}$
is satisfiable.  The graph $G_\mathcal{C}$\footnote{To improve
  readability, we draw all graphs in this paper from left to right,
  instead of upwards.}  is roughly sketched in
Fig.~\ref{fig:overallview}.

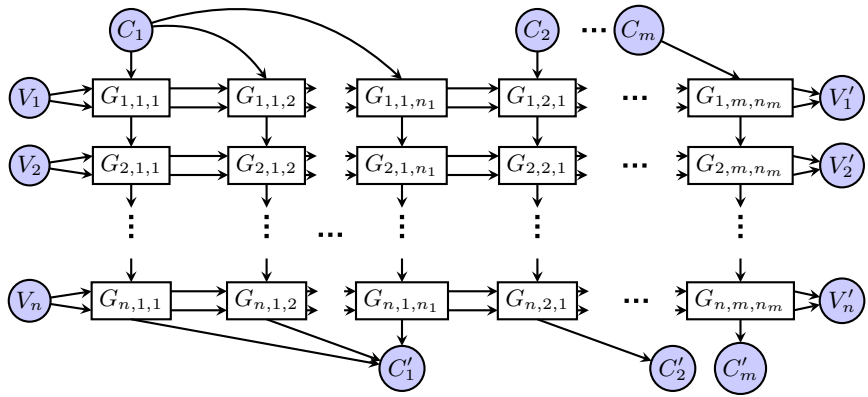
\begin{figure}
  \centering
  \begin{tikzpicture}[scale=0.9, thick,middle
  node/.style={circle,draw,fill=blue!20},inner sep = 0.04cm,outer
  sep=0, minimum size = 0.0cm,special
  node/.style={fill=white,rectangle,draw,thick,inner sep = 0.1cm},
  >=stealth]

\node[middle node] (v1) at (-7.5,1){$V_1$};
\node[middle node] (v2) at (-7.5,0){$V_2$};
\node[middle node] (v3) at (-7.5,-2){$V_n$};

    \node[middle node] (c1) at (-6,2){$C_1$};
    \node[special node] (g111) at (-6,1){$G_{1,1,1}$};
    \node[special node] (g121) at (-6,0){$G_{2,1,1}$};
    \node[special node] (g1n1) at (-6,-2){$G_{n,1,1}$};

    \node[special node] (g112) at (-4,1){$G_{1,1,2}$};
    \node[special node] (g122) at (-4,0){$G_{2,1,2}$};
    \node[special node] (g1n2) at (-4,-2){$G_{n,1,2}$};
    
        \node[special node] (g11n1) at (-2,1){$G_{1,1,n_1}$};
    \node[special node] (g12n1) at (-2,0){$G_{2,1,n_1}$};
    \node[special node] (g1nn1) at (-2,-2){$G_{n,1,n_1}$};
    
    \node[middle node] (c2) at (0,2){$C_2$};
        \node[special node] (g211) at (0,1){$G_{1,2,1}$};
    \node[special node] (g221) at (0,0){$G_{2,2,1}$};
    \node[special node] (g2n1) at (0,-2){$G_{n,2,1}$};
    
    \node[middle node] (c3) at (1.5,2){$C_m$};
        \node[special node] (gm1nm) at (3,1){$G_{1,m,n_m}$};
    \node[special node] (gm2nm) at (3,0){$G_{2,m,n_m}$};
    \node[special node] (gmnnm) at (3,-2){$G_{n,m,n_m}$};

\node[middle node] (v1c) at (4.5,1){$V'_1$};
\node[middle node] (v2c) at (4.5,0){$V'_2$};
\node[middle node] (v3c) at (4.5,-2){$V'_n$};

\node[middle node] (c1c) at (-2,-3){$C'_1$};
\node[middle node] (c2c) at (2,-3){$C'_2$};
\node[middle node] (c3c) at (3,-3){$C'_m$};

\draw[->,thick]([yshift=1pt]v1)--([yshift=4pt]g111.west);
\draw[->,thick]([yshift=-1pt]v1)--([yshift=-4pt]g111.west);

\draw[->,thick]([yshift=4pt]v2)--([yshift=4pt]g121.west);
\draw[->,thick]([yshift=-4pt]v2)--([yshift=-4pt]g121.west);

\draw[->,thick]([yshift=4pt]v3)--([yshift=4pt]g1n1.west);
\draw[->,thick]([yshift=-4pt]v3)--([yshift=-4pt]g1n1.west);

\draw[->,thick]([yshift=4pt]g111.east)--([yshift=4pt]g112.west);
\draw[->,thick]([yshift=-4pt]g111.east)--([yshift=-4pt]g112.west);

\draw[->,thick]([yshift=4pt]g121.east)--([yshift=4pt]g122.west);
\draw[->,thick]([yshift=-4pt]g121.east)--([yshift=-4pt]g122.west);

\draw[->,thick]([yshift=4pt]g1n1.east)--([yshift=4pt]g1n2.west);
\draw[->,thick]([yshift=-4pt]g1n1.east)--([yshift=-4pt]g1n2.west);

\draw[->,thick]([yshift=4pt]g112.east)--([yshift=4pt,xshift=5pt]g112.east);
\draw[->,thick]([yshift=-4pt]g112.east)--([yshift=-4pt,xshift=5pt]g112.east);

\draw[->,thick]([yshift=4pt]g122.east)--([yshift=4pt,xshift=5pt]g122.east);
\draw[->,thick]([yshift=-4pt]g122.east)--([yshift=-4pt,xshift=5pt]g122.east);

\draw[->,thick]([yshift=4pt]g1n2.east)--([yshift=4pt,xshift=5pt]g1n2.east);
\draw[->,thick]([yshift=-4pt]g1n2.east)--([yshift=-4pt,xshift=5pt]g1n2.east);

\draw[->,thick]([yshift=4pt,xshift=-5pt]g11n1.west)--([yshift=4pt]g11n1.west);
\draw[->,thick]([yshift=-4pt,xshift=-5pt]g11n1.west)--([yshift=-4pt]g11n1.west);

\draw[->,thick]([yshift=4pt,xshift=-5pt]g12n1.west)--([yshift=4pt]g12n1.west);
\draw[->,thick]([yshift=-4pt,xshift=-5pt]g12n1.west)--([yshift=-4pt]g12n1.west);

\draw[->,thick]([yshift=4pt,xshift=-5pt]g1nn1.west)--([yshift=4pt]g1nn1.west);
\draw[->,thick]([yshift=-4pt,xshift=-5pt]g1nn1.west)--([yshift=-4pt]g1nn1.west);

\draw[->,thick]([yshift=4pt]g11n1.east)--([yshift=4pt]g211.west);
\draw[->,thick]([yshift=-4pt]g11n1.east)--([yshift=-4pt]g211.west);

\draw[->,thick]([yshift=4pt]g12n1.east)--([yshift=4pt]g221.west);
\draw[->,thick]([yshift=-4pt]g12n1.east)--([yshift=-4pt]g221.west);

\draw[->,thick]([yshift=4pt]g1nn1.east)--([yshift=4pt]g2n1.west);
\draw[->,thick]([yshift=-4pt]g1nn1.east)--([yshift=-4pt]g2n1.west);
\draw[->,thick]([yshift=4pt]g211.east)--([yshift=4pt,xshift=5pt]g211.east);
\draw[->,thick]([yshift=-4pt]g211.east)--([yshift=-4pt,xshift=5pt]g211.east);

\draw[->,thick]([yshift=4pt]g221.east)--([yshift=4pt,xshift=5pt]g221.east);
\draw[->,thick]([yshift=-4pt]g221.east)--([yshift=-4pt,xshift=5pt]g221.east);

\draw[->,thick]([yshift=4pt]g2n1.east)--([yshift=4pt,xshift=5pt]g2n1.east);
\draw[->,thick]([yshift=-4pt]g2n1.east)--([yshift=-4pt,xshift=5pt]g2n1.east);
\draw[->,thick]([yshift=4pt,xshift=-5pt]gm1nm.west)--([yshift=4pt]gm1nm.west);
\draw[->,thick]([yshift=-4pt,xshift=-5pt]gm1nm.west)--([yshift=-4pt]gm1nm.west);

\draw[->,thick]([yshift=4pt,xshift=-5pt]gm2nm.west)--([yshift=4pt]gm2nm.west);
\draw[->,thick]([yshift=-4pt,xshift=-5pt]gm2nm.west)--([yshift=-4pt]gm2nm.west);

\draw[->,thick]([yshift=4pt,xshift=-5pt]gmnnm.west)--([yshift=4pt]gmnnm.west);
\draw[->,thick]([yshift=-4pt,xshift=-5pt]gmnnm.west)--([yshift=-4pt]gmnnm.west);

\draw[->,thick]([yshift=4pt]gm1nm.east)--(v1c.170);
\draw[->,thick]([yshift=-4pt]gm1nm.east)--(v1c.190);

\draw[->,thick]([yshift=4pt]gm2nm.east)--(v2c.170);
\draw[->,thick]([yshift=-4pt]gm2nm.east)--(v2c.190);

\draw[->,thick]([yshift=4pt]gmnnm.east)--(v3c.170);
\draw[->,thick]([yshift=-4pt]gmnnm.east)--(v3c.190);

\draw[->,thick](c1)--(g111.north);
\draw (c1)edge[bend left=30,->,thick](g112.north);
\draw[->,thick](c1)edge[bend left=30,->,thick](g11n1.north);

\draw[->,thick](g1n1.south)--(c1c);
\draw[->,thick](g1n2.south)--(c1c);
\draw[->,thick](g1nn1.south)--(c1c);

\draw[->,thick](c2)--(g211.north);
\draw[->,thick](g2n1.south)--(c2c);

\draw[->,thick](c3)--(gm1nm.north);
\draw[->,thick](gmnnm.south)--(c3c);
\draw[->,thick](g111.south)--(g121.north);
\draw[->,thick](g112.south)--(g122.north);
\draw[->,thick](g11n1.south)--(g12n1.north);
\draw[->,thick](g211.south)--(g221.north);
\draw[->,thick](gm1nm.south)--(gm2nm.north);

\draw[->,thick](g121.south)--([yshift=-10pt]g121.south);
\draw[->,thick](g122.south)--([yshift=-10pt]g122.south);
\draw[->,thick](g12n1.south)--([yshift=-10pt]g12n1.south);
\draw[->,thick](g221.south)--([yshift=-10pt]g221.south);
\draw[->,thick](gm2nm.south)--([yshift=-10pt]gm2nm.south);

\draw[ultra thick,dotted]([yshift=-13pt]g121.south)--([yshift=-23pt]g121.south);
\draw[ultra thick,dotted]([yshift=-13pt]g122.south)--([yshift=-23pt]g122.south);
\draw[ultra thick,dotted]([yshift=-13pt]g12n1.south)--([yshift=-23pt]g12n1.south);
\draw[ultra thick,dotted]([yshift=-13pt]g221.south)--([yshift=-23pt]g221.south);
\draw[ultra thick,dotted]([yshift=-13pt]gm2nm.south)--([yshift=-23pt]gm2nm.south);

\draw[->,thick]([yshift=10pt]g1n1.north)--(g1n1.north);
\draw[->,thick]([yshift=10pt]g1n2.north)--(g1n2.north);
\draw[->,thick]([yshift=10pt]g1nn1.north)--(g1nn1.north);
\draw[->,thick]([yshift=10pt]g2n1.north)--(g2n1.north);
\draw[->,thick]([yshift=10pt]gmnnm.north)--(gmnnm.north);

\draw[ultra thick, dotted]([yshift=20pt,xshift=-35pt]g1nn1.north)--([yshift=20pt,xshift=-25pt]g1nn1.north);

\draw[ultra thick, dotted]([xshift=20pt]g211.east)--([xshift=30pt]g211.east);
\draw[ultra thick, dotted]([xshift=20pt]g221.east)--([xshift=30pt]g221.east);
\draw[ultra thick, dotted]([xshift=20pt]g2n1.east)--([xshift=30pt]g2n1.east);

\draw[ultra thick, dotted]([xshift=10pt]c2.east)--([xshift=20pt]c2.east);

  \end{tikzpicture}
  
  \caption{Structure of the graph $G_\mathcal{C}$}
  \label{fig:overallview}
\end{figure}

We will have the source/target pairs $(V_i, V_i')\in V^2$ for $i\in
[n]$ and $(C_j, C_j')\in V^2$ for $j \in [m]$, as well as some other
source/target pairs inside the gadgets $G_{i,j,t}$ that guarantee
further properties. As the picture suggests, there will be two
possible paths from $V_i$ to $V_i'$, an upper path and a lower path
and our construction will ensure that these paths cannot
interleave. Any interpretation of the variable $V_i$ will thus
correspond to the choice of a unique path from $V_i$ to
$V_i'$. Furthermore, we will ensure that there is a path from $C_j$ to
$C_j'$ if and only if some literal is interpreted such that $C_j$ is
satisfied under this interpretation.

We need some additional gadgets which we describe first to simplify
the presentation of the main proof. All missing proofs can be found in
the appendix. 

\smallskip

\noindent\textbf{Routing Gadget}: The r\^ole of a routing gadget
is to act as a planar routing device. It has two incoming connections,
the edges $e_t$ from the top and $e_l$ from the left, and two outgoing
connections, the edges $e_b$ to the bottom and $e_r$ to the right.
The gadget is constructed in a way that in any solution to the
disjoint paths problem it allows for only two ways of routing a path
through the gadget, either using $e_t$ and $e_b$ or $e_l$ and $e_r$.

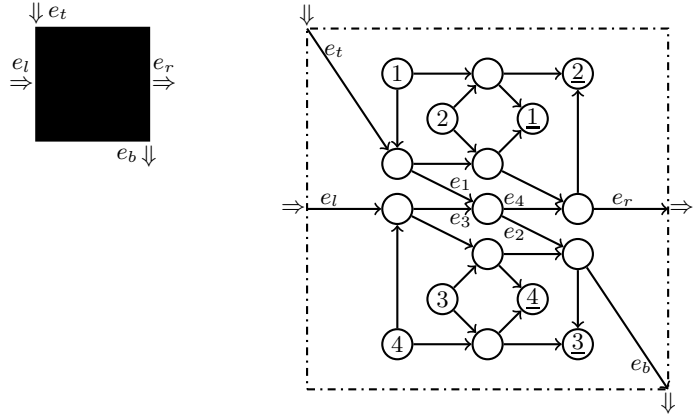
\begin{figure}
  \centering
  \begin{tikzpicture}[=>stealth, thick,mn/.style={circle,draw,fill=white,inner sep = 0.02cm,minimum size=0.4cm}]
\draw[fill=black] (-6,4) rectangle (-4.5,2.5);
\node (inup1) at (-6,4.2){\textbf{$\Downarrow$}};
\node (in1) at ([xshift=2pt]inup1.east){$e_t$};
\node(inrit1) at (-6.2,3.25){\textbf{$\Rightarrow$}};
\node (in2) at ([yshift=2pt]inrit1.north){$e_l$};
\node (outrit1) at (-4.3,3.25){$\Rightarrow$};
\node (out1) at ([yshift=2pt]outrit1.north){$e_r$};
\node (outdown1) at (-4.5,2.3){$\Downarrow$};
\node (out2) at ([xshift=-2pt]outdown1.west){$e_b$};

\begin{scope}[scale=0.6, yshift=2.65cm]
\draw[style={dash pattern=on 3pt off 2pt on \the\pgflinewidth off 2pt}] (-4,4) rectangle (4,-4);

\node (inup) at (-4,4.3){\textbf{$\Downarrow$}};
\node(inrit) at (-4.3,0){\textbf{$\Rightarrow$}};

\node[mn] (m1) at (-2,3){$1$};
\node[mn] (dc1) at (-2,1){};
\node[mn] (dc2) at (-2,0){};
\node[mn] (m2) at (-2,-3){$4$};

\node[mn] (m3) at (-1,2){2};
\node[mn] (m4) at (-1,-2){3};

\node[mn] (dc3) at (0,3){};
\node[mn] (dc4) at (0,1){};
\node[mn] (dc5) at (0,0){};
\node[mn] (dc6) at (0,-1){};
\node[mn] (dc7) at (0,-3){};

\node[mn] (m5) at (1,2){\underline{$1$}};
\node[mn] (m6) at (1,-2){\underline{$4$}};

\node[mn] (m7) at (2,3){\underline{$2$}};
\node[mn] (dc8) at (2,0){};
\node[mn] (dc9) at (2,-1){};
\node[mn] (m8) at (2,-3){\underline{$3$}};

\node (outrit) at (4.3,0){$\Rightarrow$};
\node (outdown) at (4,-4.3){$\Downarrow$};

\node at (-3.4,3.5){$e_{t}$};
\node at (-3.5,0.2){$e_{l}$};
\node at (-0.6,0.55){$e_1$};
\node at (-0.6,-0.25){$e_3$};
\node at (0.6,-0.55){$e_2$};
\node at (0.6,0.2){$e_4$};
\node at (3,0.2){$e_r$};
\node at (3.4,-3.5){$e_b$};

\draw [->] (-4,4)--(dc1);

\draw [->] (-4,0)--(dc2);

\draw [->] (dc1)--(dc5);
\draw [->] (m1)--(dc1);
\draw [->] (dc1)--(dc4);
\draw [->] (dc2)--(dc6);
\draw [->] (dc2)--(dc5);
\draw [->] (m2)--(dc2);
\draw [->] (m2)--(dc7);
\draw [->] (m1)--(dc3);

\draw [->] (m3)--(dc3);
\draw [->] (m3)--(dc4);
\draw [->] (m4)--(dc6);
\draw [->] (m4)--(dc7);

\draw [->] (dc3)--(m5);
\draw [->] (dc3)--(m7);
\draw [->] (dc4)--(m5);
\draw [->] (dc4)--(dc8);
\draw [->] (dc5)--(dc8);
\draw [->] (dc5)--(dc9);
\draw [->] (dc6)--(dc9);
\draw [->] (dc6)--(m6);
\draw [->] (dc7)--(m6);
\draw [->] (dc7)--(m8);

\draw [->] (dc8)--(m7);
\draw [->] (dc8)--(4,0);
\draw [->] (dc9)--(m8);
\draw [->] (dc9)--(4,-4);
\end{scope}

\end{tikzpicture}
  \caption{The routing gadget. In the following, when a routing
    gadget appears as a subgadget in a figure, it will be
    represented by a black box as shown on the left.}
  \label{fig:blackbox}%
\end{figure}

\noindent Formally, the gadget is defined as the graph displayed in
Fig.~\ref{fig:blackbox} with source/tar\-get pairs $(i, \underline{i})$
for $i \in [4]$.  Immediately from the construction of the gadget we
get the following lemma which captures the properties of routing
gadgets needed in the sequel.

\begin{lemma}\label{lem:routing}
  Let $R$ be a routing gadget.
  \begin{enumerate}
  \item There is a solution of the disjoint paths problem in $R$.
  \item Let $P_1, \dots, P_4$ be any solution to the disjoint paths
    problem in $R$, where $P_i$ links vertex $i$ to
    $\underline{i}$. Let $H := R\setminus \bigcup_{i=1}^4 P_i$. Then
    $H$ does not contain a path which goes through $e_t$ to $e_r$ or
    through $e_l$ to $e_b$ but there are paths $P,P'$ in~$H$ such
    that $P$ goes through $e_t$ to $e_b$ and $P'$ goes through $e_l$
    to $e_r$.
  \item There are no two disjoint paths $P, P'$ in $G$ such that $P$
    contains $e_l$ and $e_r$ and $P'$ contains $e_t$ and $e_b$.
  \end{enumerate}
\end{lemma}

\smallskip

\noindent\textbf{Crossing Gadget}:
A crossing gadget has two incoming connections to its left via the
vertices $H^{in}$ and $L^{in}$ and two outgoing connections to its
right via the vertices $H^{out}$ and $L^{out}$. Furthermore, it has
one incoming connection at the top via the vertex $T$ and outgoing
connection at the bottom via the vertex $B$. Intuitively, we want that
in any solution to the disjoint paths problem, there is exactly one
path $P$ going from left to right and exactly one path $P'$ going from
top to bottom. Furthermore, if $P$ enters the gadget via $H^{in}$ then
it should leave it via $H^{out}$ and if it enters the gadget via
$L^{in}$ then it should leave it via $L^{out}$. Of course, in a planar
graph there cannot be such disjoint paths $P, P'$ as they must cross
at some point. We will have to split one of the paths, say $P$, by
removing the outward source/sink pair and introducing two new
source/sink pairs, one to the left of $P'$ and one to its right.

\begin{figure}
  \label{fig:upwardplanargadget}
  \begin{tikzpicture}[scale=0.9, =>stealth,thick,middle
    node/.style={circle,draw,fill=blue!20},inner sep = 0.06cm, minimum size = 0.0cm,black
    box/.style={fill=black, inner sep = 0.1cm}, special
    node/.style={fill=white,rectangle,draw,very thick,inner sep = 0.1cm}]

    \draw[style={dash pattern=on 3pt off 2pt on \the\pgflinewidth off
      2pt}] (-6,4) rectangle (6.5,-4);

    \node[special node,minimum width=1] (hin) at (-6,2){$H^{in}$};
    \node[special node,minimum width=1] (lin) at (-6,-2){$L^{in}$};
	
    \node[special node] (t) at (-5,4){$T$}; 
    \node[middle node] (m1) at (-5,3){$m_1$}; 
    \node[black box] (b1) at (-5,2){\textcolor{white}{$b_1$}};
	
    \node[middle node] (m2) at (-4,2) {$m_2$};
    \node[special node] (x) at (-4,0){$X$};
    \node[black box] (b2) at (-4,-1){\textcolor{white}{$b_2$}};
    \node[middle node] (m3) at (-4,-2){$m_3$};
    
    \node[middle node] (m4) at (-2,2){$m_4$};
    \node[black box] (b3) at (-2,1){\textcolor{white}{$b_3$}};
    \node[special node] (w) at (-2,0){$W$};
    \node[middle node] (m5) at (-2,-2){$m_5$};
    
    \node[black box] (b4) at (0,2){\textcolor{white}{$b_4$}};
    \node[middle node] (m0) at (0,0) {$m_0$};
    \node[middle node] (m6) at (0,-2){$m_6$};
    
    \node[middle node] (m7) at (2,2){$m_7$};
    \node[special node] (z) at (2,0){$Z$};  
    \node[black box] (b5) at (2,-1){\textcolor{white}{$b_5$}};
    \node[middle node] (m8) at (2,-2){$m_8$};
    
    \node[middle node] (m9) at (4,2){$m_9$};
    \node[black box] (b6) at (4,1){\textcolor{white}{$b_6$}};
    \node[special node] (y) at (4,0){$Y$};
    \node[middle node] (m10) at (4,-2){$m_{10}$};
    
    \node[middle node] (m11) at  (5,-2){$m_{11}$};
    \node[middle node] (m12) at (5,-3){$m_{12}$};
    \node[special node] (b) at (5, -4){$B$};

    \node[special node] (hout) at (6.5,2){$H^{out}$};	
    \node[special node] (lout) at (6.5,-2){$L^{out}$};
    
    \draw[->,very thick](hin)--(b1);
    \draw[->,very thick](lin)--(m3);
    
    \draw[->,very thick](t)--(m1);
    \draw[->,very thick](m1)--(b1);
    \node at (-5.25,2.55){\textbf{$e^-$}};
    \node at (-3,3.3){\textbf{$e^+$}};
    \draw[->,very thick](m1)--(0,3)--(b4);
    \draw[->,very thick](b1)--(m2);
    \draw[->,very thick](b1)--(-5,-1)--(b2);
    
    \draw[->,very thick](m2)--(m4);
    \draw[->,very thick](x)--(m2);
    \draw[->,very thick](x)--(b2);
    \draw[->,very thick](b2)--(m3);
    \draw[->,very thick](b2)--(-3,-1)--(-3,1)--(b3);
    \draw[->,very thick](m3)--(m5);
    
    \draw[->,very thick](m4)--(b3);
    \draw[->,very thick](m4)--(b4);
    \draw[->,very thick](b3)--(w);
    \draw[->,very thick](b3)--(-1,1)--(-1,0)--(m0);
    \draw[->,very thick](m5)--(w);
    \draw[->,very thick](m5)--(m6);
    
    \draw[->,very thick](b4)--(m0);
    \draw[->,very thick](b4)--(m7);
    \draw[->,very thick](m0)--(m6);
    \draw[->,very thick](m0)--(1,0)--(1,-1)--(b5);
    \draw[->,very thick](m6)--(m8);
    \draw[->,very thick](m6)--(0,-3)--(m12);
    
    \draw[->,very thick](m7)--(m9);
    \draw[->,very thick](z)--(m7);
    \draw[->,very thick](z)--(b5);
    \draw[->,very thick](b5)--(3,-1)--(3,1)--(b6);
    \draw[->,very thick](b5)--(m8);
    \draw[->,very thick](m8)--(m10);
    
    \draw[->,very thick](m9)--(hout);
    \draw[->,very thick](m9)--(b6);
    \draw[->,very thick](b6)--(5,1)--(m11);
    \draw[->,very thick](b6)--(y);
    \draw[->,very thick](m10)--(y);
    \draw[->,very thick](m10)--(m11);
    
    \draw[->,very thick](m11)--(m12);
    \draw[->,very thick](m11)--(lout);
    \draw[->,very thick](m12)--(b);
\end{tikzpicture}
  \caption{The crossing gadget}
\end{figure}
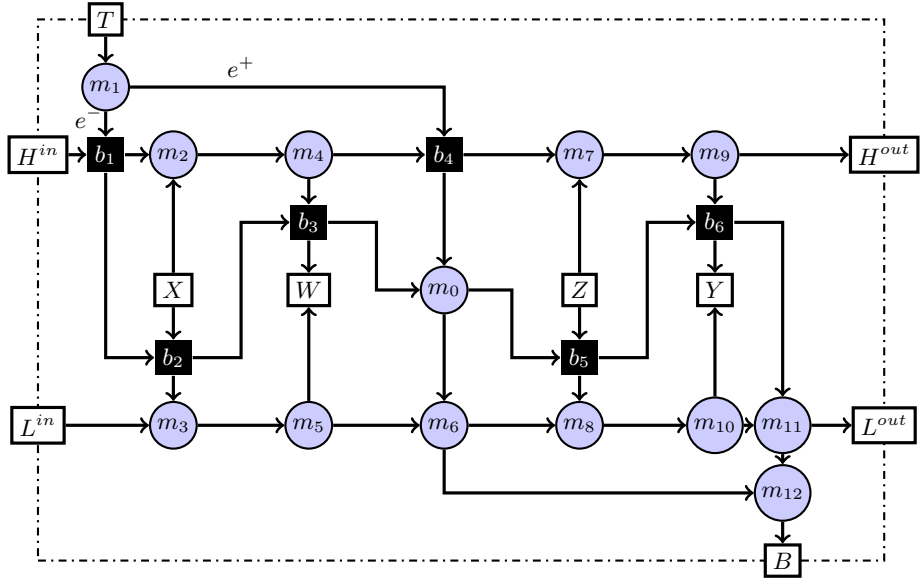

Formally, the gadget is defined as the graph displayed
in Fig.~\ref{fig:upwardplanargadget}. The following lemma follows
easily from Lemma~\ref{lem:routing} (See Appendix~\ref{proof:crossinggadget} for proof).

\begin{lemma}\label{lem:crossinggadget}
  Let $G$ be a crossing gadget.
  \begin{enumerate}
  \item There are uniquely determined vertex disjoint paths $P_1$
    from $H^{in}$ to $W$, $P_2$ from $T$ to $B$ and $P_3$ from $X$
    to $Y$. Let $H:=G\setminus \bigcup_{i=1}^{3} P_i$. Then $H$
    contains a path from $Z$ to $H^{out}$ but it does not contain a
    path from $Z$ to $L^{out}$.
  \item There are uniquely determined vertex disjoint paths $Q_1$
    from $L^{in}$ to $W$, $Q_2$ from $T$ to $B$ and $Q_3$ from $X$
    to $Y$. Let $H:=G\setminus \bigcup_{i=1}^{3} Q_i$. Then $H$
    contains a path from $Z$ to $L^{out}$ but it does not contain a
    path from $Z$ to $H^{out}$.
  \end{enumerate}
\end{lemma}
The next lemma shows that we can connect crossing gadgets in rows in
a useful way. It follows easily by induction from
Lemma~\ref{lem:crossinggadget}. 

Let $G_1,\ldots, G_s$ be a sequence of crossing gadgets drawn from
left to right in in that order. We address the inner vertices of the
gadgets by their names in the gadget equipped with corresponding
subscripts, e.g., we write $H_1^{in}$ for the vertex $H^{in}$ of
gadget $G_1$. For each $j \in [s-1]$, we add the edges $(H_j^{out},
H_{j+1}^{in})$ and $(L_j^{out}, L_{j+1}^{in})$ and call the
resulting graph a \emph{row of crossing gadgets}. We equip this
graph with the source/target pairs $(X_j, Y_{j}), (Z_j, W_{j+1})$
for $j \in [s-1]$ to obtain an associated vertex disjoint paths
problem $\mathcal{P}_r$ (the subscript $r$ stands for row). Denote
by $\mathcal{P}_r^+$ the problem $\mathcal{P}_r$ with additional
source/target pair $(H_1^{in}, W_1)$ and by $\mathcal{P}_r^-$ the
problem $\mathcal{P}_r$ with additional source/target pair
$(L_1^{in}, W_1)$.

\begin{lemma}\label{lem:row}
  Let $G$ be a row of crossing gadgets. Then both associated vertex
  disjoint paths problems $\mathcal{P}_r^+$, $\mathcal{P}_r^-$ have
  unique solutions. Each path in the solution of $\PPP_r^+$ from
  $Z_i$ to $W_{i+1}$ passes through $H_{i+1}^{in}$ and each path in
  the solution of $\PPP_r^-$ from $Z_i$ to $W_{i+1}$ passes through
  $L_{i+1}^{in}$.
\end{lemma}
Proof is in Appendix~\ref{proof:row}. The next lemma shows that we can force a relation between rows and
columns of crossing gadgets. 

Let $G_1,\ldots, G_t$ be a sequence of crossing gadgets drawn from top
to bottom in that order. For each $i \in [t-1]$, we add the edge
$(B_{i}, T_{i+1})$ and call the resulting graph a \emph{column of
  crossing gadgets}. We equip this graph with the source/target pairs
$(X_i, Y_i)$ for $i \in [t]$ and with the pair $(T_1, B_t)$ to obtain
an associated vertex disjoint paths problem $\mathcal{P}$.

\begin{lemma}\label{lem:column}
  Let $G$ be a column of crossing gadgets. Let $P_1,\ldots, P_t$ be
  a sequence of vertex disjoint paths such that $P_i$ connects
  either $H^{in}_i$ or $L^{in}_i$ to $W_i$. Let
  $H:=G\setminus\bigcup_{i=1}^{t}P_i$.
  \begin{enumerate}
  \item The vertex disjoint paths problem $\mathcal{P}$ on $H$ has a
    solution.
  \item There is a unique path $Q$ connecting $T_1$ to $B_t$ which
    uses edge $e^+$ in $G_i$ if and only if $P_i$ starts at $H_i^{in}$
    and the edge $e^-$ in $G_i$ if and only if $P_i$ starts at
    $L_i^{in}$.
  \end{enumerate}
\end{lemma}
\noindent Note that such paths $P_i$ as stated in the lemma exist and
they are uniquely determined by Lemma~\ref{lem:crossinggadget}.

We are now ready to construct a vertex disjoint paths instance for
any SAT instance $\mathcal{C}$. 
  
\begin{definition}
  \label{dfn:sat-instance}
  Let $\CCC$ be a SAT instance over the variables
  $\mathcal{V}=\{V_1,\ldots, V_n \}$ and let $\{C_1,…,C_m \}$ be its
  set of clauses. For $j \in [m]$ let $C_j=\{L_{j,1}$, $L_{j,2},…$,
  $L_{j,n_j}\}$, where each $L_{j,s}$ is a literal, i.e.,~a variable
  or the negation thereof.
  
  \begin{enumerate}
  \item The graph $G_\mathcal{C}$ is defined as follows.
    \begin{itemize}
    \item For each variable $V\in\mathcal{V}$ we introduce two
      vertices $V$ and $V'$. 
    \item For each clause $C\in\CCC$ we introduce two vertices $C$
      and $C'$.
    \item For each variable $V_i$ and each literal $L_{j,t}$ in
      clause $j$ we introduce a crossing gadget $G_{i,j,t}$.
    \item For $i\in [n]$ we add the edges $(V_i,H_{i,1,1}^{in})$,
      $(V_i, L_{i,1,1}^{in})$, $(H_{i,m,n_m}^{out}, V_i')$ and
      $(L_{i,m,n_m}^{out}, V_i')$.
    \item For $j \in [m], t\in [n_j]$ we add the edges $(C_j,
      T_{1,j,t})$ and $(B_{n,j,t}, C_j')$
    \item Finally, we delete the edge $e^+$ for all $i\in [n], j \in
      [m], t\in [n_j]$ in $G_{i,j,t}$ if $L_{j,t}$ is a variable the
      edge $e^-$ if it is a negated variable.
    \end{itemize}
    We draw the graph $G_\mathcal{C}$ as shown in Fig.~\ref{fig:overallview}.
    
  \item We define the following vertex disjoint paths problem
    $\mathcal{P}_\mathcal{C}$ on $G_\mathcal{C}$. We add all
    source/target pairs that are defined inside the routing gadgets. Furthermore:
    \begin{itemize}
    \item For $i \in [n], j \in [m], t\in [n_j-1]$, we add the
      pairs
      \begin{itemize}
      \item $(V_i,W_{i,1,1})$,
      \item $(Z_{i,m,n_m}, V_i')$,
      \item $(X_{i,j,t},Y_{i,j,t})$ and
      \item $(Z_{i,j,t}, W_{i,j,t+1})$
      \end{itemize}
    \item For $i\in [n], j \in [m-1]$, we add the pairs
      $(Z_{i,j,n_j},W_{i,j+1,1})$.
    \item For $j \in [m]$, we add the pairs $(C_j, C_j')$.
    \end{itemize}
    
  \end{enumerate}
\end{definition}

\noindent The proof of the following theorem is based on the fact that
in our construction, edge $e^+$ is present in gadget $G_{i,j,t}$, if
and only if $C_j$ does not contain variable $V_i$ negatively and $e^-$
is present in gadget $G_{i,j,t}$, if and only if $C_j$ does not
contain variable $V_i$ positively (especially, both edges are present
if the clause does not contain the variable at all). In particular,
every column contains exactly one gadget where one edge is
missing. Now it is easy to conclude with Lemma~\ref{lem:row} and
Lemma~\ref{lem:column}. We defer the formal proof to the appendix.
 
\begin{theorem}
  \label{thm:sat-vdpp}
  Let $\mathcal{C}$ be a SAT-instance and let
  $\mathcal{P}_\mathcal{C}$ be the corresponding vertex disjoint paths
  instance on $G_\mathcal{C}$ as defined in
  Definition~\ref{dfn:sat-instance}. Then $\mathcal{C}$ is satisfiable
  if and only if $\mathcal{P}_\mathcal{C}$ has a solution.
\end{theorem}

It is easily seen that the presented reduction can be computed in
polynomial time and this finishes the proof of Theorem~\ref{thm:np}.

If we replace the vertices $C_i$ and $C_i'$ with directed paths, then
it is easy to convert the graph $G_\CCC$ to a directed grid graph,
i.e., a subgraph of the infinite grid. This implies that the problem
is NP-complete even on upward planar graphs of maximum degree $4$.

\section{A Linear Time Algorithm for Fixed $k$}

In this section we prove that the $k$-disjoint paths
problem for upward planar digraphs can be solved
in linear time for any fixed value of $k$. In other words, the problem
is fixed-parameter tractable by a linear time parameterized
algorithm. 

\begin{theorem}\label{thm:k-vdpp}
  The problem \UPDP can be solved in time $\Oof(k!\cdot k\cdot n)$,
  where $n := |V(G)|$.
\end{theorem}

For the rest of the section we fix a planar upward graph $G$ together
with an upward planar embedding and $k$ pairs $(s_1, t_1), \dots,
(s_k, t_k)$ of vertices. We will not distinguish notationally between
$G$ and its upward planar embedding.  Whenever we speak about a vertex
$v$ on a path $P$ we mean a vertex $v\in V(G)$ which is contained in
$P$. If we speak about a \emph{point on the path} we mean a point $(x,
y) \in \R^2$ which is contained in the drawing of $P$ with respect to
the upward planar drawing of $G$.  The algorithm is based on the
concept of a path in $G$ being \emph{to the right} of another path
which we define next.

\begin{definition}\label{def:path-right}
  Let $P$ be a path in an upward planar drawing of $G$.  Let $(x, y)$
  and $(x', y')$ be the two endpoints of $P$ such that $y \leq y'$,
  i.e.~$P$ starts at $(x, y)$ and ends at $(x', y')$. We define
  \[
    \begin{array}{rcl}
      \Right{P} & := & \{ (u, v) \in \R^2\mathrel : y \leq v \leq y' \text{ and $u'< u$
        for all $u'$ such that }     (u', v) \in P\}\\
      \Left{P} & := & \{ (u, v)  \in \R^2 \mathrel : y \leq v \leq y' \text{ and $u'
        > u$   for all $u'$ such that }
      (u', v) \in P\}.
    \end{array}
  \]
\end{definition}

The next two lemmas follow immediately from the definition of upward
planar drawings.

\begin{lemma}\label{lem:path-right}
  Let $P$ and $Q$ be vertex disjoint paths in an upward planar drawing of $G$. Then
  either $\Right{P} \cap Q = \emptyset$ or $\Left{P}\cap Q = \emptyset$.
\end{lemma}

\begin{lemma}\label{lem:up-path-points}
  Let $P$ be a directed path in an upward planar drawing of a
  digraph~$G$.  For $i=1,2,3$ let $p_i := (x_i, y_i)$ be distinct
  points in $P$ such that $y_1 < y_2 < y_3$. Then $p_1, p_2, p_3$
  occur in this order on $P$.
\end{lemma}

\begin{definition}
  \label{dfn:right}
  Let $P$ and $Q$ be two vertex disjoint paths in $G$. 
  \begin{enumerate}
  \item A point $p = (x, y) \in \R^2\setminus P$ is to the \emph{right}
    of $P$ if $p\in \Right{P}$. Analogously, we say that $(x, y) \in \R^2\setminus P$ is
    to the \emph{left} of $P$ if $p\in \Left{P}$.
  \item The path $P$ is to the \emph{right} of $Q$, denoted by $Q\prec
    P$ if there exists a
    point $p\in P$ which to the right of some point $q\in Q$. We write
    $\prec^*$ for the transitive closure of $\prec$.
   \item If $\PPP$ is a set of pairwise disjoint paths in $G$, we
     write $\prec_\PPP$ and $\prec^*_\PPP$ for the restriction of
     $\prec$ and $\prec^*$, resp., to the paths in $\PPP$. 
  \end{enumerate}
\end{definition}

We show next that for every set $\PPP$ of pairwise vertex disjoint
paths in $G$ the relation $\prec^*$ is a partial order on $\PPP$.
Towards this aim, we first show that $\prec$ is irreflexive and
anti-symmetric on $\PPP$.


\begin{lemma}\label{lem:antisymmetric}
  Let $\PPP$ be a set of pairwise disjoint paths in $G$.
  \begin{enumerate}
  \item The relation $\prec_\PPP$ is irreflexive.
  \item The relation $\prec_\PPP$ is anti-symmetric, i.e.~if
    $P_1\prec_\PPP P_2$ then $P_2\not\prec_\PPP P_1$ for any $P_1, P_2\in\PPP$.
  \end{enumerate}
\end{lemma}
\begin{proof}
  The first claim immediately follows from the definition of
  $\prec$. Towards the second statement, suppose there are $P_1, P_2\in
  \PPP$ such that $P_1\prec_\PPP P_2$ and $P_2\prec_\PPP P_1$. 

  Hence, for $j=1,2$ and $i=1,2$ there are points $p^i_j = (x^i_j,
  y^i_j)$ such that $p^i_j \in P_i$ and $x^1_1 < x^2_1$, $y^1_1 =
  y^2_1$ and $x^1_2 > x^2_2$, $y^1_2 = y^2_2$. W.l.o.g.~we assume that
  $y^1_1 < y^1_2$. Let $Q \subseteq P$ be the subpath of $P$ from
  $p^1_1$ to $p^1_2$, including the endpoints. Let $Q_1 := \{ (x^1_1,
  z) \mathrel : z < y^1_1 \}$ and $Q_2 := \{ (x^1_2, z) \mathrel : z >
  y^1_2 \}$ be the two lines parallel to the $y$-axis going from
  $p^1_1$ towards negative infinity and from $p^1_2$ towards
  infinity. Then $Q_1 \cup Q\cup Q_2$ separates the plane into two
  disjoint regions $R_1$ and $R_2$ each containing a point of
  $P_2$. As $P_1$ and $P_2$ are vertex disjoint but $p^2_1$ and
  $p^2_2$ are connected by $P_2$, $P_2$ must contain a point in $Q_1$
  or $Q_2$ which, on $P_2$ lies between $p^2_1$ and $p^2_2$. But the
  $y$-coordinate of any point in $Q_1$ is strictly smaller than
  $y^2_1$ and $y^2_2$ whereas the $y$-coordinate of any point in $Q_2$
  is strictly bigger than $y^2_1$ and $y^2_2$. This contradicts
  Lemma~\ref{lem:up-path-points}.
\end{proof}

We use the previous lemma to show that $\prec^*_\PPP$ is a partial
order for all sets $\PPP$ of pairwise vertex disjoint paths.

\begin{lemma}\label{thm:partialorder}
  Let $\PPP$ be a set of pairwise vertex disjoint directed paths. Then
  $\prec^*_\PPP$ is a partial order.
\end{lemma}
\begin{proof}
  By definition, $\prec^*_\PPP$ is transitive. Hence we only need to
  show that it is anti-symmetric for which, by transitivity, it suffices to show that
  $\prec^*_\PPP$ is irreflexive.
  
  To show that $\prec^*_\PPP$ is irreflexive, we prove by induction on
  $k$ that if $P_0, \dots, P_k \in \PPP$ are paths such that $P_0
  \prec_\PPP \dots \prec_\PPP P_k$ then $P_k \not \prec_\PPP
  P_0$.  As for all $P\in \PPP$, $P\not\prec_\PPP P$, this proves the
  lemma.
  
  Towards a contradiction, suppose the claim was false and let $k$ be
  minimum such that there are paths 
   $P_0, \dots, P_{k} \in \PPP$ with $P_0
  \prec_\PPP \dots \prec_\PPP P_{k}$ and $P_k \prec_\PPP P_0$. 
  By Lemma~\ref{lem:antisymmetric}, $k>1$.

  Let $R := \bigcup_{i=0}^{k-2}\Right{P_i}$. Note that $k-2 \geq 0$,
  so $R$ is not empty.  Furthermore, as for all $P, Q$ with $P\prec
  Q$, $\Right{P} \cap \Right{Q} \not=\emptyset$, $R$ is a connected
  region in $\R^2$ without holes.  Let $L := \bigcup_{i=1}^{k-1}
  \Left{P_i}$. Again, as $k>1$, $L\not=\emptyset$ and $L$ is a
  connected region without holes.

  As $P_{k-2} \prec_\PPP P_{k-1}$, we
  have $L \cap R \not=\emptyset$ and therefore $L
  \cup R$ separates the plane into two unbounded regions, the upper
  region $T$ and the lower region $B$.

  The minimality of $k$ implies that $P_i \not\prec_\PPP P_k$ for all
  $i<k-1$ and therefore $R  \cap P_k =  \emptyset$. 
  Analogously, as $P_k \not \prec_\PPP P_i$ for any $i>0$,  we have $L \cap P_k = \emptyset$.
  Hence, either $P_k \subseteq B$ or $P_k\subseteq T$. W.l.o.g.~we
  assume $P_k\subseteq B$. We will show that $\Left{P_0} \cap B  = \emptyset$.


  Suppose there was a point $(x, y) \in P$ and some $x'< x$ such that
  $(x', y) \in B$. This implies that $y < v$ for all $(u,v) \in
  L$. But this implies that $B$ is bounded by $\Right{P_0}$ and $L$
  contradicting the fact that $\Right{P_{k-1}}\cap B \not=\emptyset$.
\end{proof}

We have shown so far that $\prec^*$ is a partial order on every set of
pairwise vertex disjoint paths. 

\begin{remark}
  Note that if two paths $P, Q\in\PPP$ are incomparable with respect to
  $\prec^*_\PPP$ then one path is strictly above the other,
  i.e.~$(\Right{P} \cup \Left{P}) \cap (\Right{Q}\cup \Left{Q}) =
  \emptyset$. This is used in the next lemma.
\end{remark}

\begin{definition}
  Let $s, t\in V(G)$ be vertices in $G$ such that there is a directed
  path from $s$ to $t$. The \emph{right-most}
  $s$-$t$-path in $G$ is an $s$-$t$-path $P$ such that for all
  $s$-$t$-paths $P'$, $P\subseteq P'\cup \Right{P'}$. 
\end{definition}

\begin{lemma}\label{lem:path-comb}
  Let $s,t \in V(G)$ be two vertices and let $P$ be a path from $s$ to
  $t$ in an upward planar drawing of $G$. If $P'$ is an $s$-$t$ path
  such that $P'\cap \Right{P} \not=\emptyset$ then there is an $s$-$t$
  path $Q$ such that $Q\subseteq P\cup\Right{P}$ and $Q\cap \Right{P} \not=\emptyset$.
\end{lemma}
\begin{proof}
  If $P'\subseteq P\cup\Right{P}$ we can take $Q=P'$. Otherwise,
  i.e.~if $P'\cap\Left{P} \not=\emptyset$, then as the graph is planar
  this means that $P$ and $P'$ share internal vertices. In this case
  we can construct $Q$ from $P\cup P'$ where for subpaths of $P$ and
  $P'$ between two vertices in
  $P\cap P'$ we always take the subpath to the right. 
\end{proof}

\begin{corollary}\label{cor:path-comb}
  Let $s, t\in V(G)$ be vertices in $G$ such that there is a directed
  path from $s$ to $t$. Then there is a unique \emph{right-most}
  $s$-$t$-path in $G$.
\end{corollary}

The corollary states that between any two $s$ and $t$, if there is an
$s$-$t$ path then there is a rightmost one. The proof of
Lemma~\ref{lem:path-comb} also indicates how such a path can be
computed. This is formalised in the next lemma.

\begin{lemma}\label{lem:lin-time}
  There is a linear time algorithm which, given an upward planar
  drawing of a graph $G$ and two vertices $s, t\in V(G)$ computes the
  right-most $s$-$t$-path in $G$, if such a path exists.
\end{lemma}
\begin{proof}
  We first use a depth-first search starting at $s$ to compute the set
  of vertices $U\subseteq V(G)$ reachable from $s$. Clearly, if
  $t\not\in U$ then there is no $s$-$t$-path and we can
  stop. Otherwise we use a second, inverse depth-first search to
  compute the set $U'\subseteq U$ of vertices from which $t$ can be
  reached. Finally, we compute the right-most $s$-$t$ path inductively by
  starting at $s \in U'$ and always choosing the right-most successor
  of the current vertex until we reach $t$. The right-most successor
  is determined by the planar embedding of $G$. As $G$ is acyclic,
  this procedure produces the right-most path and can clearly be
  implemented in linear time.
\end{proof}

We show next that in any solution $\PPP$ to the disjoint paths problem in an
upward planar digraph, if $P\in \PPP$ is a maximal element with
respect to $\prec^*_{\PPP}$, we can replace $P$ by the right-most
$s$-$t$ path and still get a valid solution, where $s$ and $t$ are the
endpoints of $P$.

\begin{lemma}\label{lem:valid-solution}
  Let $G$ be an upward planar graph with a fixed upward planar
  embedding and let $(s_1, t_1), \dots, (s_k, t_k)$ be pairs of
  vertices. Let $\PPP$ be a set of pairwise disjoint paths connecting
  $(s_i, t_i)$ for all $i$. Let $P\in \PPP$ be path connecting $s_i$
  and $t_i$, for some $i$, which is maximal with respect
  to $\prec^*_\PPP$. Let $P'$ be the right-most $s_i-t_i$-path in
  $G$. 
  Then $\PPP\setminus \{P\} \cup \{ P'\}$ is also a valid
  solution to the disjoint paths problem on $G$ and  $(s_1, t_1),
  \dots, (s_k, t_k)$. 
\end{lemma}
\begin{proof}
  All we have to show is that $P'$ is disjoint from all $Q\in
  \PPP\setminus \{ P\}$. Clearly, as $P$ and $P'$ are both upward
  $s_i$-$t_i$ paths, we have $P'\subseteq P \cup \Left{P} \cup \Right{P}$.

  By the remark above, if $P$ and $Q$ are
  incomparable with respect to $\prec^*_\PPP$, then one is above the
  other and therefore $Q$ and $P'$ must be disjoint. Now suppose $Q$
  and $P$ are comparable and therefore $Q\prec^*_\PPP P$. This implies
  that $(P \cup \Right{P}) \cap Q = \emptyset$ and therefore $Q\cap P'=\emptyset$.
\end{proof}

The previous lemma yields the key to the proof of
Theorem~\ref{thm:k-vdpp} which is provided in
Appendix~\ref{proof:k-vdpp}.

\smallskip

We remark that we can easily extend this result to ``almost upward
planar'' graphs, i.e., to graphs such that the deletion of at most $h$
edges yields an upward planar graph. As finding an upward planar drawing
of an upward planar graph is NP-complete, this might be of use if we
have an approximation algorithm that produces almost upward planar
embeddings.



\section{Conclusion}

In this paper we showed that the $k$-vertex disjoint paths problem is
NP-complete on a restricted and yet very interesting class of planar
digraphs. 
On the other hand, we provided a fast algorithm to approach this hard problem by finding good partial order.
It is an interesting question to
investigate whether the $k!$ factor in the running time of our
algorithm can be improved. Another direction of research is to extend
our result to more general but still restricted graph classes, such as
to digraphs embedded on a torus such that all edges are monotonically
increasing in the $z$-direction or to acyclic planar graphs.

\bibliographystyle{plain}
\bibliography{papers}

\pagebreak
\appendix
\textbf{\Huge{Appendix}}
\smallskip
\section{NP-Completeness Proof}

\label{proof:crossinggadget}
\begin{proof}[Proof of Lemma~\ref{lem:crossinggadget}]
\begin{enumerate}
\item In the first case, the only possible paths are the following. 
	\begin{itemize}
		\item $H^{in}\rightarrow b_1\rightarrow m_2\rightarrow m_4\rightarrow b_3\rightarrow W$
		\item $X\rightarrow b_2\rightarrow m_3\rightarrow m_5\rightarrow m_6\rightarrow m_8\rightarrow m_{10}\rightarrow Y$
		\item $Z\rightarrow m_7\rightarrow m_9\rightarrow H^{out}$
		\item $T\rightarrow m_1\rightarrow b_4\rightarrow m_0\rightarrow b_5\rightarrow b_6\rightarrow m_{11}\rightarrow m_{12}\rightarrow B$
	\end{itemize}
	\item In the second case, the only possible paths are the following. 
	\begin{itemize}
	\item	$L^{in}\rightarrow m_3\rightarrow m_5\rightarrow W$
	\item $X\rightarrow m_2\rightarrow m_4\rightarrow b_4\rightarrow m_7\rightarrow m_9\rightarrow b_6\rightarrow Y$
	\item $Z\rightarrow b_5\rightarrow m_8\rightarrow m_{10}\rightarrow m_{11}\rightarrow L^{out}$
	\item $T\rightarrow m_1\rightarrow b_1\rightarrow b_2\rightarrow b_3\rightarrow m_0\rightarrow m_6\rightarrow m_{12}\rightarrow B$
	\end{itemize}
\end{enumerate}
\end{proof}

\label{proof:row}
\begin{proof}[Proof of Lemma~\ref{lem:row}]
By Lemma~\ref{lem:crossinggadget} we know that if we use $H_i^{in}$($L_i^{in}$) in each gadget then we going out by the $H_i^{out}$($L_i^{out}$), this simply follows that if we have a row of consecutive gadgets then we have exactly one of a following structure between gadgets:
\begin{enumerate}
\item First gadget uses $H_1^{in}$ then : $H_1^{in}\dots H_1^{out}\rightarrow H_2^{in}\dots H_2^{out}\dots H_s^{in}\dots H_s^{out}$
\item First gadget uses $L_1^{in}$ then : $L_1^{in}\dots L_1^{out}\rightarrow L_2^{in}\dots L_2^{out}\dots L_s^{in}\dots L_s^{out}$
\end{enumerate}
\end{proof}

\label{proof:sat-vdpp}
\begin{proof}[Proof of Theorem~\ref{thm:sat-vdpp}]
Let $i\in [n]$. If $\beta(V_i)=1$, we consider the problem
 $\mathcal{P}_r^+$ associated with row $i$ as defined in
 Lemma~\ref{lem:row} and if $\beta(V_i)=0$, we consider $\PPP_r^-$
 associated with row $i$. By the lemma, either problem has a unique
 solution and we can easily identify this solution with a solution to
 the subproblem of $\PPP_\CCC$ restricted to row $i$. At this point,
 we have constructed disjoint paths for all pairs but for $(C_j,
 C_j')$, $j\in [m]$.

 By construction, $e^+$ is present in gadget $G_{i,j,t}$, if and only
 if $C_j$ does not contain variable $V_i$ negatively and $e^-$ is
 present in gadget $G_{i,j,t}$, if and only if $C_j$ does not contain
 variable $V_i$ positively (especially, both edges are present if the
 clause does not contain the variable at all). Note that every column
 contains exactly one gadget where one edge is missing.

 As $\beta$ is satisfying, for each $j\in [m]$, there is $t\in [n_j]$ 
 such that $\beta(L_{j,t})=1$. Assume without loss of
 generality that $L_{j,t}=V_i$ for some $i\in [n]$. Consider the
 column which corresponds to $L_{j,t}$. With all the above paths
 fixed, we are in the situation of Lemma~\ref{lem:column}. By
 construction, both edges $e^+$ and $e^-$ are present in any gadgets
 $G_{t,j,t}$ where $V_t\neq L_{j,t}$. But also $e^+$ is present in
 $G_{i,j,t}$, as $V_i$ occurs positively in $C_j$. Hence by
 Lemma~\ref{lem:column} we find a path disjoint to all of the above
 paths from $T_{1,j,t}$ to $B_{n,j,t}$ and this path can be extended
 to a path from $C_j$ to $C_j'$.

 \medskip
 \noindent Now assume that there is a solution to the disjoint paths
 problem in $G_\mathcal{C}$. Obviously, each path $P_i$ connecting
 $V_i$ to $V_i'$ uses exactly one of $H_{i,1,1}^{in}$ or
 $L_{i,1,1}^{in}$. We define $\beta:\{V_1,\ldots,
 V_n\}\rightarrow\{0,1\}$ by $\beta(V_i)=1$ if and only if $P_i$ uses
 $H_{i,1,1}^{in}$. By Lemma~\ref{lem:column}, each path $Q_k$
 connecting $C_t$ to $C_t'$ can only pass through $G_{i,j,t}$ if it
 is consistent with this assignment. Hence, $C_k$ contains a variable
 $V_i$ with $\beta(V_i)=1$.
\end{proof}
\smallskip

\section{Main Algorithm and Correctness Proof}
\begin{proof}[Proof of Theorem~\ref{thm:k-vdpp}]
\label{proof:k-vdpp}
  Let $G$ with an upward planar drawing of $G$ and $k$ pairs $(s_1,
  t_1), \dots, (s_k, t_k)$ be given. To decide whether there is a
  solution to the disjoint paths problem on this instance we proceed
  as follows. In the first step we compute for each $s_i$ the set of
  vertices reachable from $s_i$. If for some $i$ this does not include
  $t_i$ we reject the input as obviously there cannot be any
  solution. 

  In the second step, for every possible permutation $\pi$ of $\{ 1,
  \dots, k\}$ we proceed as follows.
  Let $i_1 := \pi(k), \dots, i_k := \pi(1)$ be the numbers $1$ to $k$
  ordered as indicated by $\pi$ and let $u_j := s_{i_j}$ and $v_j :=
  t_{i_j}$, for all $j\in [k]$. We can view $\pi$ as a linear order
  on $1, \dots, k$ and for every such $\pi$ we will search for a
  solution $\PPP$ of the disjoint paths problem for which
  $\prec^*_\PPP$ is consistent with $\pi$. 

  For a given $\pi$ as above we inductively construct a 
  sequence $\PPP_0, \dots, \PPP_k$ of sets of pairwise vertex disjoint
  paths such that for all $i$, $\PPP_i$ contains a set of $i$ paths
  $P_1, \dots, P_i$ such that for all $j\in [i]$ $P_j$ links $u_{j}$ to $v_{j}$. We
  set $\PPP_0 := \emptyset$ which obviously satisfies the
  condition. Suppose for some $0\leq i < k$, $\PPP_i$ has 
  already been constructed. To obtain $\PPP_{i+1}$ we compute the
  right-most path linking $u_{i+1}$ to $v_{i+1}$ in the graph
  $G\setminus \bigcup \PPP_i$. By Lemma~\ref{lem:lin-time}, this can
  be done in linear time for each such pair $(s_{i+1}, t_{i+1})$. If
  there is such a path $P$ we define $\PPP_{i+1} := \PPP_i \cup \{ P
  \}$. Otherwise we reject the input. 
  Once we reach $\PPP_k$ we stop and output $\PPP_k$ as solution. 

  Clearly, for every permutation $\pi$ the algorithm can be
  implemented to run in time $\Oof(k\cdot n)$, using
  Lemma~\ref{lem:lin-time}, so that the total running time is
  $\Oof(k!\cdot k\cdot n)$ as required. 

  Obviously, if the algorithm outputs a set $\PPP$ of disjoint paths
  then $\PPP$ is indeed a solution to the problem. What is left to
  show is that whenever there is a solution to the disjoint
  path problem, then the algorithm will find one. 
  
  So let $\PPP$ be a solution, i.e.~a set of $k$ paths $P_1, \dots,
  P_k$ so that $P_i$ links $s_i$ to $t_i$. Let $\leq$ be a linear
  order on $\{ 1, \dots, k\}$ that extends $\prec^*_\PPP$ and let
  $\pi$ be the corresponding permutation such that $(u_1, v_1), \dots,
  (u_k, v_k)$ is the ordering of $(s_1, t_1), \dots, (s_k, t_k)$
  according to $\leq$. We claim that for this permutation $\pi$ the
  algorithm will find a solution. 
  Let $P$ be the right-most $u_k$-$v_k$-path in $G$ as computed by the
  algorithm. By Lemma~\ref{lem:valid-solution}, $\PPP \setminus \{ P_k
  \} \cup P$ is also a valid solution so we can assume that $P_k =
  P$. Hence, $P_1, \dots, P_{k-1}$ form a solution of the disjoint
  paths problem for $(u_1, v_1), \dots,  (u_{k-1}, v_{k-1})$ in
  $G\setminus P$. By repeating this argument we get a solution
  $\PPP':= \{ P'_1, \dots, P'_k\}$ such that each $P'_i$ links $u_i$ to
  $v_i$ and is the right-most $u_i$-$v_i$-path in $G\setminus
  \bigcup_{j>i} P'_j$. But this is exactly the solution found by the
  algorithm. This prove the correctness of the algorithm and concludes
  the proof of the theorem.
\end{proof}








\end{document}